\documentclass[10pt, conference]{IEEEtran}
\usepackage{subfigure}
\usepackage{setspace}
\usepackage{amsmath}
\usepackage{amssymb}
\usepackage{amsfonts}
\usepackage{amscd}
\usepackage{mathrsfs}
\usepackage[final]{graphicx}
\usepackage{graphicx}
\usepackage{psfrag}
\usepackage{epsfig}
\usepackage{color}
\usepackage{url}
\usepackage{textcomp}
\usepackage{multirow}
\input{epsf.sty}

\newtheorem{lemma}{Lemma}
\newtheorem{definition}{Definition}

\begin{document}

\title{Reduced ML-Decoding Complexity, Full-Rate STBCs for 4 Transmit Antenna Systems}
\author{
\authorblockN{K. Pavan Srinath}
\authorblockA{Dept of ECE, Indian Institute of science \\
Bangalore 560012, India\\
Email:pavan@ece.iisc.ernet.in\\
}
\and
\authorblockN{B. Sundar Rajan}
\authorblockA{Dept of ECE, Indian Institute of science \\
Bangalore 560012, India\\
Email:bsrajan@ece.iisc.ernet.in\\
}
}
\maketitle
\begin{abstract}
For an $n_t$ transmit, $n_r$ receive antenna system ($n_t \times n_r$ system), a {\it{full-rate}} space time block code (STBC) transmits $min(n_t,n_r)$ complex symbols per channel use. In this paper, a scheme to obtain a full-rate STBC for 4 transmit antennas and any $n_r$, with reduced ML-decoding complexity is presented.  The weight matrices of the proposed STBC are obtained from the unitary matrix representations of Clifford Algebra. By puncturing the symbols of the STBC, full rate designs can be obtained for $n_r < 4$. For any value of $n_r$, the proposed design offers the least ML-decoding complexity among known codes. The proposed design is comparable in error performance to the well known perfect code for 4 transmit antennas while offering lower ML-decoding complexity. Further, when $n_r < 4$, the proposed design has higher ergodic capacity than the punctured Perfect code. Simulation results which corroborate these claims are presented.
\end{abstract}

\section{Introduction and Background}
Complex orthogonal designs (CODs) \cite{TJC}, \cite{TiH}, although provide linear Maximum Likelihood (ML) decoding, do not offer a high rate of transmission. A full-rate code for an $n_t \times n_r$ MIMO system transmits $min(n_t,n_r)$ complex symbols per channel use. Among the CODs, only the Alamouti code for 2 transmit antennas is full-rate for a $2 \times 1$ MIMO system. A full-rate STBC can efficiently utilize all the degrees of freedom the channel provides. Generally, an increase in the rate also results in an increase in the ML-decoding complexity. The Golden code \cite{BRV} for 2 transmit antennas is an example of a full-rate STBC for any number of receive antennas. Until recently, the ML-decoding complexity of the Golden code was known to be of the order of $M^4$, where $M$ is the size of the signal constellation. However, it was shown in \cite{SrR_arxiv}, \cite{john_barry1} that the Golden code has a decoding complexity of the order of $M^2\sqrt{M}$ only. A lot of attention is being given to reducing the ML-complexity of full-rate codes. Current research focuses on obtaining high rate codes with reduced ML-decoding complexity (refer to Sec. \ref{sec2} for a formal definition), since high rate codes are essential to exploit the available degrees of freedom of the MIMO channel. For 2 transmit antennas, the Silver code \cite{HTW}, \cite{PGA}, is a full-rate code with full-diversity and an ML-decoding complexity of order $M^2$. For 4 transmit antennas, Biglieri et. al. proposed a rate-2 STBC which has an ML-decoding complexity of $M^4\sqrt{M}$ for square QAM without full-diversity \cite{BHV}. It was, however, shown that there was no significant reduction in error performance at low and medium SNR when compared with the then best known code - the DjABBA code \cite{HTW}. This code was obtained by multiplexing Quasi-orthogonal designs (QOD) for 4 transmit antennas \cite{JH}. 
Currently, the best full-rate STBC for $4\times2$ system with full diversity and an ML-decoding complexity of $M^4\sqrt{M}$ is the one given in \cite{SrR_arxiv}. This code was obtained by multiplexing the coordinate interleaved orthogonal designs (CIODs) for 4 transmit antennas \cite{ZS}. These results show that codes obtained by multiplexing low complexity STBCs can result in high rate STBCs with reduced ML-decoding complexity and without any significant degradation in the error performance when compared with the best existing STBCs. Such an approach has also been adopted in \cite{Robert} to obtained high rate codes from multiplexed orthogonal designs.

The well known full-rate STBC for 4 transmit antennas is the Perfect code \cite{ORBV}. It is full-diversity, full-rate, information-lossless and DMG optimal. On the negative side, its ML-decoding complexity is of the order of $M^{16}$. By puncturing the layers of the Perfect code, one can obtain full-rate designs for $n_r < 4$ receive antennas. However, for 2 receive antennas, Biglieri's code and  the code proposed in \cite{SrR_arxiv} beat the punctured Perfect code (puncturing refers to replacing certain symbols with zeros. Since the Perfect code has 4 layers, for $n_r < 4$, the symbols corresponding to $ 4-n_r$ layers are punctured.) in error performance \cite{BHV}, \cite{SrR_arxiv}, while having lower ML-decoding complexity as well. It is natural to ask if such similar advantages can be obtained for $n_r > 2$ receive antennas. In this paper, we answer this question in the affirmative by generalizing the result in \cite{SrR_arxiv} to any number of receive antennas. In particular, the contributions of this paper are:

\begin{enumerate}
\item We propose a full-rate STBC for 4 transmit antennas and any value of $n_r$. This is done by successively constructing full-rate STBCs for 2,3 and 4 receive antennas. The weight matrices of the STBCs are obtained from the unitary matrix representations of Clifford Algebras. For any $n_r$, the ML-decoding complexity of the proposed STBC is lower than that of the Perfect code by a factor of $M^3$ for non-regular QAM constellations (for square QAM, it is lower by a factor of $M$).
\item Like the Perfect code, the proposed code is information lossless for $n_r \geq 4$ receive antennas, while for lower number of receive antennas, the proposed code has higher ergodic capacity than the punctured Perfect code.
\item The proposed code has full-diversity and a better symbol error performance than the punctured Perfect code for 2 receive antennas at any SNR, while for the 3 and 4 receive antennas, although not a full-diversity STBC, its performance is similar to that of the perfect code in the low and medium SNR range. The reason for this is explained in Subsection \ref{subsec1}.

\end{enumerate}

The paper is organized as follows. In Section \ref{sec2}, we present the system model and the relevant definitions. The proposed code is presented in Section \ref{sec3} and the ML-decoding complexity and the ergodic capacity issues are discussed in Section \ref{sec4}. Simulation results are discussed in Section \ref{sec5} and the concluding remarks are made in Section \ref{sec6}.

\textit{\textbf{Notations}:} Throughout, bold, lowercase letters are used to denote vectors and bold, uppercase letters are used to denote matrices. Let $\textbf{X}$ be a complex matrix. Then, $\textbf{X}^{\dagger}$ denotes the Hermitian or the transpose of $\textbf{X}$, depending on whether $\textbf{X}$ is complex or real, resp., and $j$ represents $\sqrt{-1}$. The set of all real and complex numbers are denoted by $\mathbb{R}$ and $\mathbb{C}$, respectively. The real and the imaginary part of a complex number $x$ are denoted by $x_I$ and $x_Q$, respectively. $\Vert \textbf{X} \Vert$ denotes the Frobenius norm of $\textbf{X}$, and $\textbf{I}_T$ and $\textbf{O}_T$ denote the $T\times T$ identity matrix and the null matrix, respectively. The Kronecker product is denoted by $\otimes$. For a complex random variable $X$, $\mathcal{E}[X]$ denotes the mean of $X$.

For a complex variable $x$, the $\check{(.)}$ operator acting on $x$ is defined as follows.
\begin{equation*}
\check{x} \triangleq \left[ \begin{array}{rr}
                             x_I & -x_Q \\
                             x_Q & x_I \\
                            \end{array}\right].
\end{equation*}
The $\check{(.)}$ can similarly be applied to any matrix $\textbf{X} \in \mathbb{C}^{n \times m}$ by replacing each entry $x_{ij}$ by $\check{x}_{ij}$, $i=1,2,\cdots, n, j = 1,2,\cdots,m$ , resulting in a matrix denoted by $\check{\textbf{X}} \in \mathbb{R}^{2n \times 2m}$.

Given a complex vector $\textbf{x} = [ x_1, x_2, \cdots, x_n ]^T$, $\tilde{\textbf{x}}$ is defined as
\begin{equation*}
\tilde{\textbf{x}} \triangleq [ x_{1I},x_{1Q}, \cdots, x_{nI}, x_{nQ} ]^T.
\end{equation*}

\section{System Model}
\label{sec2}
We consider Rayleigh block fading MIMO channel with full channel state information (CSI) at the receiver but not at the transmitter. For $n_t \times n_r$ MIMO transmission, we have
\begin{equation}\label{Y}
\textbf{Y} = \sqrt{\frac{SNR}{n_t}}\textbf{HS + N}
\end{equation}

\noindent where $\textbf{S} \in \mathbb{C}^{n_t \times T}$ is the codeword matrix whose average energy is given by $\mathcal{E}(\Vert \textbf{S} \Vert^2) = n_tT$, transmitted over $T$ channel uses, $\textbf{N} \in \mathbb{C}^{n_r \times T}$ is a complex white Gaussian noise matrix with i.i.d entries $\sim
\mathcal{N}_{\mathbb{C}}\left(0,1\right)$ and $\textbf{H} \in \mathbb{C}^{n_r\times n_t}$ is the channel matrix with the entries assumed to be i.i.d circularly symmetric Gaussian random variables $\sim \mathcal{N}_\mathbb{C}\left(0,1\right)$. $\textbf{Y} \in \mathbb{C}^{n_r \times T}$ is the received matrix.

\begin{definition}\label{def1}$\left(\textbf{Code rate}\right)$ Code rate is the average number of independent information symbols transmitted per channel use. If there are $k$ independent complex information symbols in the codeword which are transmitted over $T$ channel uses, then, the code rate is $k/T$ complex symbols per channel use.
\end{definition}

\begin{definition}\label{def2}$\left(\textbf{Full-rate STBCs}\right)$ For an $n_t \times n_r$ MIMO system, if the code rate is $min\left(n_t,n_r\right)$, then the STBC is said to be \emph{\textbf{full-rate}}.
\end{definition}

 Assuming ML-decoding, the ML-decoding metric that is to be minimized over all possible values of codewords $\textbf{S}$ is given by
 \begin{equation}
\label{ML}
 \textbf{M}\left(\textbf{S}\right) = \Vert \textbf{Y} - \sqrt{\frac{SNR}{n_t}}{}\textbf{HS} \Vert_F^2
 \end{equation}

\begin{definition}\label{def3}$\left(\textbf{ML-Decoding complexity}\right)$
The ML decoding complexity is measured in terms of the maximum number of symbols that need to be jointly decoded in minimizing the ML decoding metric.
\end{definition}
For eg., if the codeword transmits $k$ independent symbols of which a maximum of $p$ symbols need to be jointly decoded, the ML-decoding complexity is of the order of $M^{p}$, where $M$ is the size of the signal constellation. If the code has an ML-decoding complexity of order less than $M^k$, the code is said to admit \emph{\textbf{reduced ML-decoding}}.

\begin{definition}\label{def4}$\left(\textbf{Generator matrix}\right)$ For any STBC that encodes $k$ information symbols, the \textbf{\emph{generator}} matrix $\textbf{G}$ is defined by the following equation  \cite{BHV}.
\begin{equation*}
\widetilde{vec\left(\textbf{S}\right)} = \textbf{G} \tilde{\textbf{s}},
\end{equation*}
\noindent where $\textbf{S}$ is the codeword matrix, $\textbf{s} \triangleq \left[ s_1, s_2,\cdots,s_k \right]^T$ is the information symbol vector.
\end{definition}

A codeword matrix of an STBC can be expressed in terms of \textbf{\emph{weight matrices}} (linear dispersion matrices) as follows \cite{HaH}.
\begin{equation*}
\textbf{S} = \sum_{i=1}^{k}s_{iI}\textbf{A}_{2i-1}+s_{iQ}\textbf{A}_{2i}.
\end{equation*}
Here, $\textbf{A}_i,i=1,2,\cdots,2k$ are the complex weight matrices for the STBC and should form a linearly independent set over $\mathbb{R}$. It follows that
\begin{equation*}
\textbf{G} = [\widetilde{vec(\textbf{A}_1)}\ \widetilde{vec(\textbf{A}_2)}\ \cdots \ \widetilde{vec(\textbf{A}_{2k})}].
\end{equation*}

\section{Code construction using Clifford Algebra}\label{sec3}
In this section, we show how the full-rate STBC with reduced ML-decoding complexity can be constructed using unitary matrix representations of Clifford algebras. This approach was first taken in \cite{HTW} to obtain a full-rate STBC for $4 \times 4 $ MIMO systems. But here, we look at designing an STBC so that it achieves reduced ML-decoding, acceptable error performance when compared with the best existing code, which is the Perfect code and has higher ergodic capacity than the punctured Perfect code for $n_r < 4$. We construct a full-rate STBC for any number of receive antennas by using a full-rate code for 1 receive antennas to successively construct full-rate codes for 2,3 and 4 receive antennas. The design is based on the following lemma.
\begin{lemma}\label{lemma_1}
If $n=2^m$ and matrices  $\textbf{F}_1, \cdots, \textbf{F}_{2m}$, which are of size $n \times n$, anticommute pairwise, then the set of products $\textbf{F}_1^{\lambda_1}\textbf{F}_2^{\lambda_2}\cdots\textbf{F}_{2m}^{\lambda_{2m}}$ with $\lambda_i \in \{0,1\}, i = 1,2,\cdots,2m$ forms a basis for the $2^{2m}$ dimensional space of all $n \times n$ matrices over $\mathbb{C}$.
\end{lemma}
\begin{proof}
Available in \cite{anti_matric}.
\end{proof}
As a byproduct of the lemma, the set $\{\textbf{F}_1^{\lambda_1}\textbf{F}_2^{\lambda_2}\cdots\textbf{F}_{2m}^{\lambda_{2m}}$, $ j\textbf{F}_1^{\lambda_1}\textbf{F}_2^{\lambda_2}\cdots\textbf{F}_{2m}^{\lambda_{2m}}\}$ forms a basis for the $2^{2m+1}$ dimensional space of all $n \times n$ matrices over $\mathbb{R}$. We choose the matrices from this set to be weight matrices of our STBC. For $n = 4$, the following matrices (not necessarily unique), which are obtainable from the unitary matrix representations of Clifford algebra \cite{TiH}, are the 4 pairwise anticommuting matrices.
\begin{equation*}
\textbf{F}_1 = \left[ \begin{array}{cccc}
j & 0  &  0  &  0  \\
0 & -j &  0  &  0  \\
0 & 0  & -j  &  0  \\
0 & 0  &  0  &  j  \\
\end{array} \right], ~\textbf{F}_2 = \left[ \begin{array}{cccc}
     0  &   1  &   0  &   0\\
    -1  &   0  &   0  &   0\\
     0  &   0  &   0   &  1\\
     0  &   0  &  -1  &   0\\
\end{array} \right],
\end{equation*}

\begin{equation*}
\textbf{F}_3 = \left[ \begin{array}{cccc}
        0  &  j  &  0  &  0 \\
        j  &  0  &  0  &  0 \\
        0  &  0  &  0  &  j\\
        0  &  0  &  j  &  0  \\
 \end{array} \right],~\textbf{F}_4 = \left[ \begin{array}{cccc}
     0  &   0  &   1  &   0\\
     0  &   0  &   0  &  -1\\
    -1  &   0  &   0  &   0\\
     0  &   1  &   0  &   0\\
\end{array} \right].
\end{equation*}

So, $\mathcal{X} \triangleq \{\textbf{F}_1^{\lambda_1}\textbf{F}_2^{\lambda_2}\textbf{F}_{3}^{\lambda_{3}}\textbf{F}_{4}^{\lambda_{4}}, j\textbf{F}_1^{\lambda_1}\textbf{F}_2^{\lambda_2}\textbf{F}_{3}^{\lambda_{3}}\textbf{F}_{4}^{\lambda_{4}}\}$, $\lambda_i$ $\in$ $\{0,1\}$, $i = 1,2,3,4$, is the linearly independent (over $\mathbb{R}$) set of weight matrices.
Since we want reduced ML-decoding as well, the appropriate ordering of weight matrices is important. To illustrate with an example, the Silver code \cite{PGA} has 8 weight matrices corresponding to 8 real symbols (or 4 complex symbols), among which the first 4 are the weight matrices of the Alamouti code. Hence, when the last four real symbols are fixed, the first four symbols can be independently decoded. This would not have been achievable if the weight matrices were randomly allocated. So, for 4 transmit antennas, to construct a full-rate code with reduced ML-decoding complexity, we first need to construct a low decoding complexity code using some of the weight matrices from $\mathcal{X}$. For 4 transmit antennas, the best multi-group decodable code is the rate-1, single-complex symbol decodable (SSD) code which has been extensively studied in literature and is known in many forms - CIOD \cite{ZS}, MDCQOD \cite{YGT}, CUW-SSD code \cite{sanjay}. It is to be noted that all these codes have the same coding gain but different weight matrices. For our construction of a full-rate STBC, we make use of the CIOD. The codeword of the CIOD is as follows.

{\footnotesize
\begin{eqnarray}\label{ciod}
\textbf{S}_1^{ciod}(s_1,\cdots,s_4)& = & s_{1I}(\textbf{I}_4 - \textbf{F}_1\textbf{F}_2\textbf{F}_3)  + s_{1Q}(\textbf{F}_1 - \textbf{F}_2\textbf{F}_3) {}
\nonumber\\
&&{} +  s_{2I}(\textbf{F}_1\textbf{F}_3 - \textbf{F}_2)  + s_{2Q}(\textbf{F}_3 - \textbf{F}_1\textbf{F}_2) {}
\nonumber\\
&&{} + s_{3I}(\textbf{I}_4 + \textbf{F}_1\textbf{F}_2\textbf{F}_3) + s_{3Q}(\textbf{F}_1 + \textbf{F}_2\textbf{F}_3) {}
\nonumber\\
&&{} + s_{4I}(-\textbf{F}_2 - \textbf{F}_1\textbf{F}_3) + s_{4Q}(\textbf{F}_3 + \textbf{F}_1\textbf{F}_2).
\end{eqnarray}
}
In \eqref{ciod}, the symbols take values from a QAM constellation which is rotated by an angle of $(1/2)tan^{-1}2$ rad. This angle maximizes the coding gain for the code \cite{ZS}.
We can obtain a full-rate STBC for 2 receive antennas by obtaining 8 more weight matrices on post-multiplication of the weight matrices of the CIOD by $\textbf{F}_4$. This does not spoil the linear independence of the resulting set of weight matrices, which is evident from Lemma \ref{lemma_1}. So, the resulting rate-2 code has the codeword matrix as follows.
\begin{equation*}
 \textbf{S}_{2}(s_1,\cdots,s_8) = \textbf{S}_{1}^{ciod}(s_1,\cdots,s_4) +  \textbf{S}_{1}^{\prime}(s_5,\cdots,s_8)
\end{equation*}
where,
{\footnotesize
\begin{eqnarray}\label{code2}
\textbf{S}_1^{\prime}(s_5,\cdots,s_8)& =  & \textbf{S}_1^{ciod}(s_5,\cdots,s_8)\textbf{F}_4\\
& = & s_{5I}(\textbf{F}_4 - \textbf{F}_1\textbf{F}_2\textbf{F}_3\textbf{F}_4)  + s_{5Q}(\textbf{F}_1\textbf{F}_4 - \textbf{F}_2\textbf{F}_3\textbf{F}_4) {}
\nonumber\\
&&{} +  s_{6I}(\textbf{F}_1\textbf{F}_3\textbf{F}_4 - \textbf{F}_2\textbf{F}_4)  + s_{6Q}(\textbf{F}_3\textbf{F}_4 - \textbf{F}_1\textbf{F}_2\textbf{F}_4) {}
\nonumber\\
&&{} + s_{7I}(\textbf{F}_4 + \textbf{F}_1\textbf{F}_2\textbf{F}_3\textbf{F}_4) + s_{7Q}(\textbf{F}_1\textbf{F}_4 + \textbf{F}_2\textbf{F}_3\textbf{F}_4) {}
\nonumber\\
&&{} + s_{8I}(-\textbf{F}_2\textbf{F}_4 - \textbf{F}_1\textbf{F}_3\textbf{F}_4) {}
\nonumber\\
&&{} + s_{8Q}(\textbf{F}_3\textbf{F}_4 + \textbf{F}_1\textbf{F}_2\textbf{F}_4),
\end{eqnarray}
}

Note that the code whose codeword matrix is shown in \eqref{code2} is also SSD, as its weight matrices are obtained by post-multiplying the weight matrices of the CIOD by a unitary matrix, which, in this case is $\textbf{F}_4$. The rate-2 code described above does not have full-diversity. Its performance can be enhanced by using a complex scalar which results in the following codeword matrix.
\begin{equation}\label{rate2}
 \textbf{S}_{2}(s_1,\cdots,s_8) = \textbf{S}_{1}^{ciod}(s_1,\cdots,s_4) +  e^{j\pi/4}\textbf{S}_{1}^{\prime}(s_5,\cdots,s_8).
\end{equation}

In \eqref{rate2}, the use of the complex scalar $e^{j\pi/4}$ makes the code have full-diversity with a high coding gain. The value of the minimum determinant \cite{TSC} obtained for this code is 10.24 for 4-/16-QAM and this was verified by exhaustive computer search. The rate-2 code described above has the same coding gain and ML-decoding complexity as the one presented in \cite{SrR_arxiv}.

To obtain a full-rate for 3 and 4 receive antennas, we need to obtain the remaining 8 and 16 weight matrices by multiplying the weight matrices of the CIOD and the rate-2 code whose codeword matrix is given in \eqref{rate2}, respectively, by $j$. Note from Lemma \ref{lemma_1} that the above operation does not spoil the linear independence of the resulting set of weight matrices over $\mathbb{R}$. Hence, the codeword of a rate-3 code is given as follows.
\begin{equation}\label{rate3}
 \textbf{S}_{3}(s_1,\cdots,s_{12}) = \textbf{S}_{2}(s_1,\cdots,s_8) + j\textbf{S}_1^{ciod}(s_9,\cdots,s_{12})
\end{equation}
where, $\textbf{S}_{2}(s_1,\cdots,s_8)$ is given by \eqref{rate2}. The codeword matrix of a full-rate STBC for $n_r \geq 4$ is as follows.
\begin{equation}\label{rate4}
 \textbf{S}_{4}(s_1,\cdots,s_{16}) = \textbf{S}_{2}(s_1,\cdots,s_8) + j\textbf{S}_2(s_9,\cdots,s_{16}).
\end{equation}
The full rate code for $n_r \geq 4$ is given below.
 \begin{eqnarray}\label{con2}
 \textbf{S}_{4}(s_1,\cdots,s_{16})& = &\textbf{S}_{1}^{ciod}(s_1,\cdots,s_4) {}
\nonumber\\
&&{} + e^{j\pi/4}\textbf{S}_{1}^{ciod}(s_5,\cdots,s_8)\textbf{F}_4 {}
\nonumber\\
&&{} + j\textbf{S}_{1}^{ciod}(s_9,\cdots,s_{12}) {}
\nonumber\\
&&{} + je^{j\pi/4} \textbf{S}_{1}^{ciod}(s_{13},\cdots,s_{16})\textbf{F}_4.
\end{eqnarray}

Note from \eqref{con2} that the codeword matrix of the full-rate code is obtained from independent codeword matrices of 4 separate SSD codes. This property will be exploited in the next section to achieve reduced ML-decoding complexity.

\subsection{Performance of our code}\label{subsec1}
The rate-2 code whose codeword matrix is given in \eqref{rate2} has a minimum determinant of 10.24 for 4-QAM. The corresponding minimum determinant of the punctured Perfect code is 3.6304. The minimum determinants of both the codes have been calculated for 4-QAM with the average codeword energy being 16 units, i.e, $\mathcal{E} \Vert \textbf{S} \Vert^2 = n_tT$. As a result of a higher minimum determinant and hence a better coding gain, our rate-2 code is expected to perform better than the punctured Perfect code.

For 3 receive antennas, our rate-3 code whose codeword matrix is given in \eqref{rate3} does not offer full-diversity. It can be noted that the rate-3 code is obtained by multiplexing a full-diversity rate-2 code and a full-diversity rate-1 code. In other words, each codeword matrix of our code has two individual sub-codeword matrices - one sub-codeword matrix belonging to the full-diversity rate-2 code and the other belonging to the full-diversity rate-1 code. We say that our rate-3 code has two {\it embedded} full-diversity codes in it. Hence, though the rate-3 code may not have full-diversity, meaning which its minimum determinant is zero, the number of codeword difference matrices which are not full-ranked is lesser than it would be if the rate-3 code were constructed using arbitrary weight matrices. This is because there are many instances when two codewords of the rate-3 code are such that their codeword difference matrix is the same as one of the codeword difference matrices of one of the embedded full-diversity codes (This happens when the two codewords of the rate-3 code have a common sub-codeword matrix). Further, even if the two codewords of the rate-3 code do not have a common sub-codeword matrix, their difference matrix might still be full-ranked. Hence, in comparison to the number of codeword difference matrices of the rate-3 code, the number of non full-ranked codeword difference matrices is very small. A similar arguement can be done for the rate-4 code, which is full-rate for $n_r \geq 4$ and can be seen from \eqref{rate4} to have two embedded full-diversity rate-2 codes. Hence, we expect the rate-3 and the rate-4 codes to perform very well atleast in the low and medium SNR range. Simulation results presented later confirm our expectations.

{\footnotesize
\begin{figure*}
\begin{equation*}
\textbf{D}_{Perfect} =  \left[ \begin{array}{cccccccc}
x & 0 & x & 0 & x & 0 & x & 0 \\
0 & x & 0 & x & 0 & x & 0 & x \\
0 & 0 & x & 0 & x & 0 & x & 0 \\
0 & 0 & 0 & x & 0 & x & 0 & x \\
0 & 0 & 0 & 0 & x & 0 & x & 0 \\
0 & 0 & 0 & 0 & 0 & x & 0 & x  \\
0 & 0 & 0 & 0 & 0 & 0 & x & 0  \\
0 & 0 & 0 & 0 & 0 & 0 & 0 & x \\
\end{array} \right], ~~~~~\textbf{D}_{EAST}=  \left[ \begin{array}{cccccccc}
x & 0 & x & 0 & 0 & 0 & 0 & 0 \\
0 & x & 0 & x & 0 & 0 & 0 & 0 \\
0 & 0 & x & 0 & 0 & 0 & 0 & 0 \\
0 & 0 & 0 & x & 0 & 0 & 0 & 0 \\
0 & 0 & 0 & 0 & x & 0 & x & 0 \\
0 & 0 & 0 & 0 & 0 & x & 0 & x  \\
0 & 0 & 0 & 0 & 0 & 0 & x & 0  \\
0 & 0 & 0 & 0 & 0 & 0 & 0 & x \\
\end{array} \right].
\end{equation*}\hrule
\end{figure*}
}

\section{ML-Decoding Complexity and Ergodic capacity} \label{sec4}
The ML-decoding complexity of a code depends on the weight matrices of the code \cite{SrR_arxiv}. Our proposed design is such that for any number of receive antennas, reduced ML-decoding can be achieved. To see this, our code whose codeword matrix is as shown in \eqref{con2} consists of 4 multiplexed rate-1 SSD codes. This means that for $n_{min}= min(4,n_r)$, $n_{min}$ SSD codes can be multiplexed so that the code rate is $n_{min}$ complex symbols per channel use. So, for any $n_{min}$, one can fix the last $4(n_{min}-1)$ symbols and decode the first 4 symbols independently (with an additional complexity increase by a factor of only  $M$). Thus, for any number of receive antennas, the worst case ML-decoding complexity is of the order of $M^{4(n_{min}-1)+1}$. This results in a reduction in ML-decoding complexity by a factor of $M^3$ with respect to the Perfect code for general constellations.

The channel can be modelled as follows (note that $n_t = T =4$ at all places below).
\begin{equation*}
 \widetilde{vec(\textbf{Y})} = \sqrt{\frac{SNR}{n_t}}\textbf{H}_{eq}\tilde{\textbf{s}} + \widetilde{vec(\textbf{N})},
\end{equation*}
\noindent where $\textbf{H}_{eq} \in \mathbb{R}^{2n_rT\times2n_{min}T}$ is given by
\begin{equation*}
 \textbf{H}_{eq} = \left(\textbf{I}_T \otimes \check{\textbf{H}}\right)\textbf{G},
\end{equation*}
with $\textbf{G} \in \mathbb{R}^{2n_tT\times 2n_{min}T}$ being the generator matrix as in Def. \ref{def4}, so that $\widetilde{vec\left(\textbf{S}\right)} = \textbf{G} \tilde{\textbf{s}}.$ and
\begin{equation*}
\tilde{\textbf{s}} \triangleq [s_{1I},s_{1Q},\cdots,s_{(n_{min}T)I},s_{(n_{min}T)Q}]^\dagger.
\end{equation*}

When rotated QAM constellation is employed for our code, with the angle of rotation being $\theta = (1/2)tan^{-1}2$,
\begin{equation*}
 \tilde{\textbf{s}} = \textbf{F}\tilde{\textbf{x}}.
\end{equation*}
\noindent Here $\textbf{F} = \textbf{I}_{n_{min}T} \otimes \textbf{J}$, with
\begin{equation*}
 \textbf{J} \triangleq \left[\begin{array}{cc}
                     cos(\theta) & -sin(\theta)\\
                     sin(\theta) &  cos(\theta)\\
              \end{array}\right],
\end{equation*}
$\textbf{x} \triangleq [x_{1},\cdots,x_{n_{min}T}]^\dagger$, and $x_i$, $i = 1,\cdots,n_{min}T$, take values from a QAM constellation.

The ML decoding metric can now be written as
\begin{equation*}
 \textbf{M}\left(\tilde{\textbf{x}}\right) = \Vert \widetilde{vec\left(\textbf{Y}\right)} - \sqrt{\frac{SNR}{n_t}}\textbf{H}_{eq}\textbf{F}\tilde{\textbf{x}}\Vert^2.= \Vert \textbf{y}^\prime - \sqrt{\frac{SNR}{n_t}}\textbf{R}\tilde{\textbf{x}}\Vert^2,
\end{equation*}
where, $\textbf{y}^\prime = \textbf{Q}^\dagger \widetilde{vec\textbf{(Y)}}$, and on QR-decomposition, $\textbf{H}_{eq}\textbf{F} = \textbf{QR}$, $\textbf{R} \in \mathbb{R}^{2n_rT \times 2n_{min}T}$. Specifically for our design, for any number of receive antennas, the $\textbf{R}$- matrix has the following structure.
\begin{equation}\label{rmatrix}
 \textbf{R} =  \left[\begin{array}{cccc}
\textbf{D} &  \textbf{X} &  \ldots &  \textbf{X} \\
\textbf{O}_8 &  \textbf{D} &  \ldots &  \textbf{X} \\
\vdots &   \ddots  & \ddots & \vdots \\
\textbf{O}_8 & \textbf{O}_8 & \ldots & \textbf{D} \\
\end{array} \right]
\end{equation}
where $\textbf{X} \in \mathbb{R}^{8 \times 8}$ is random non-sparse matrix whose entries depend on the channel coefficients and $\textbf{D} \in \mathbb{R}^{8 \times 8}$ has the following structure.

{\footnotesize
\begin{equation*}
\textbf{D} =  \left[ \begin{array}{cccccccc}
x & x & 0 & 0 & 0 & 0 & 0 & 0 \\
0 & x & 0 & 0 & 0 & 0 & 0 & 0 \\
0 & 0 & x & x & 0 & 0 & 0 & 0 \\
0 & 0 & 0 & x & 0 & 0 & 0 & 0 \\
0 & 0 & 0 & 0 & x & x & 0 & 0 \\
0 & 0 & 0 & 0 & 0 & x & 0 & 0  \\
0 & 0 & 0 & 0 & 0 & 0 & x & x  \\
0 & 0 & 0 & 0 & 0 & 0 & 0 & x \\
\end{array} \right].
\end{equation*}
}
Here, $x$ represents a non-zero entry. The structure of the above $\textbf{R}$-matrix is due to the fact that the codeword matrix of our code is comprised of multiplexed SSD-codeword matrices. Note that for the CIOD, the entanglement between the real and the imaginary parts of a symbol is due to the constellation rotation, which is employed for full-diversity. This can also be checked by direct computation. Correspondingly, for the Perfect code, the matrix $\textbf{D}$ in \eqref{rmatrix} is as shown at the top of the next page.

\begin{table*}
\begin{center}
\begin{tabular}{|l|l|l|l|l|} \hline
 \multirow{3}{*}{$\textbf{n}_r$ (\textbf{Rx Antennas})} & \multirow{3}{*}{\textbf{code}} & \multirow{3}{*}{\textbf{Min. Determinant}} & \multicolumn{2}{c|}{\textbf{ML Decoding complexity order}}  \\ \cline{4-5}
 & &  &  \textbf{ square QAM} & \textbf{Non-rectangular}  \\
& & & & \textbf{QAM} \\ \hline
\multirow{4}{*}{2} &  DjABBA  &  0.64 &  $M^6$    &  $M^8$ \\
 &  Punctured perfect code  &  3.6304 &  $M^5\sqrt{M}$    &  $M^8$ \\
 &  EAST code \cite{barry} &  10.24 &  $M^4\sqrt{M}$    &  $M^6$ \\
 &  The proposed code  &  10.24 &  $M^4\sqrt{M}$    &  $M^5$ \\ \hline
\multirow{2}{*}{3} &  Punctured perfect code  & 0.7171  &  $M^9\sqrt{M}$     &  $M^{12}$ \\
 &  The proposed code  &  0  &  $M^8\sqrt{M}$     &  $M^9$ \\ \hline
\multirow{2}{*}{4} &  Perfect code  &  0.2269  &  $M^{13}\sqrt{M}$     &  $M^{16}$ \\
 &  The proposed code  &  0  &  $M^{12}\sqrt{M}$     &  $M^{13}$ \\ \hline
\end{tabular}
\caption{Comparison of the codes for 4 transmit antennas}
\label{table2}
\end{center}
\end{table*}

Clearly, the $\textbf{R}$-matrix for our code has more zero entries than the $\textbf{R}$-matrix of the Perfect code. This means that the interference between symbols is lesser for our code than for the Perfect code. A consequence of this is that when QAM constellations are employed, the average ML-decoding complexity using a sphere decoder \cite{ViB} is much lesser than the worst case ML-decoding complexity of $M^{4(n_{min}-1)+1}$. Note that the worst case ML-decoding complexity of our code is lower than that of the Perfect code by a factor of $M^3$ only for non-rectangular QAM constellations. But for square-QAM constellations of size $M$, where the real and the imaginary parts of a signal point can be independently decoded, the ML-decoding of our code can be reduced further by a factor of $\sqrt{M}$ (from $M^{4(n_{min}-1)+1}$ to $M^{4(n_{min}-1)}\sqrt{M}$) by quantizing (the details are presented in \cite{SrR_arxiv}) and the decoding complexity of the Perfect code can be reduced by a factor of $M^2\sqrt{M}$ (from $M^{4n_{min}}$ to $M^{(4n_{min}-3)}\sqrt{M}$). This is achieved by noting from the $\textbf{R}$-matrix structure for the Perfect code that the real parts of the symbols $s_1,s_2,s_3$ and $s_4$ can be independently decoded from imaginary parts for square QAM constellations and this reduces the complexity by a factor of $M^2$ and using quantizing further reduces the complexity by a factor of $\sqrt{M}$. Table \ref{table2} summarizes these facts.
We have also used the EAST code \cite{barry} for $4 \times 2$ MIMO for comparison with our code. The EAST code, which is full-rate of $n_r=2$ has the following $\textbf{R}$-matrix structure.
\begin{equation*}
 \textbf{R} =  \left[\begin{array}{cc}
\textbf{D} &  \textbf{X}  \\
\textbf{O}_8 &  \textbf{D}  \\
\end{array} \right],
\end{equation*}
with $\textbf{D}$ having the structure shown at the top of this page. Clearly, its worst case ML-decoding complexity order can be as low as $M^4\sqrt{M}$ for square-QAM and $M^6$ for non-rectangular QAM constellations.

It was shown in \cite{HTW} that a reduction in interference between symbols leads to a better mutual information. The ergodic capacity with the use of a space time code is given as follows \cite{JJK}.
\begin{equation}
 \mathcal{C} = \frac{1}{2T}\mathcal{E}_{\textbf{H}}log[det(\textbf{I}_{2n_rT} + \frac{SNR}{n_t}\textbf{H}_{eq}\textbf{H}_{eq}^\dagger)]
\end{equation}
Since our code has lesser interference between symbols than the Perfect code, as is evident from the $\textbf{R}$-matrix structure (hence, this is reflected in $\textbf{H}_{eq}$), its ergodic capacity is expected to be better for 2 and 3 receive antennas and this is confirmed in Fig. \ref{fig_cap}. For, $n_r \geq 4$, it can be checked that the generator matrix for our code is unitary, like that of the Perfect code. Hence, for $n_r \geq 4$, our code is information lossless, like the Perfect code.

\begin{figure*}
\centering
\includegraphics[width=7in,height=4in]{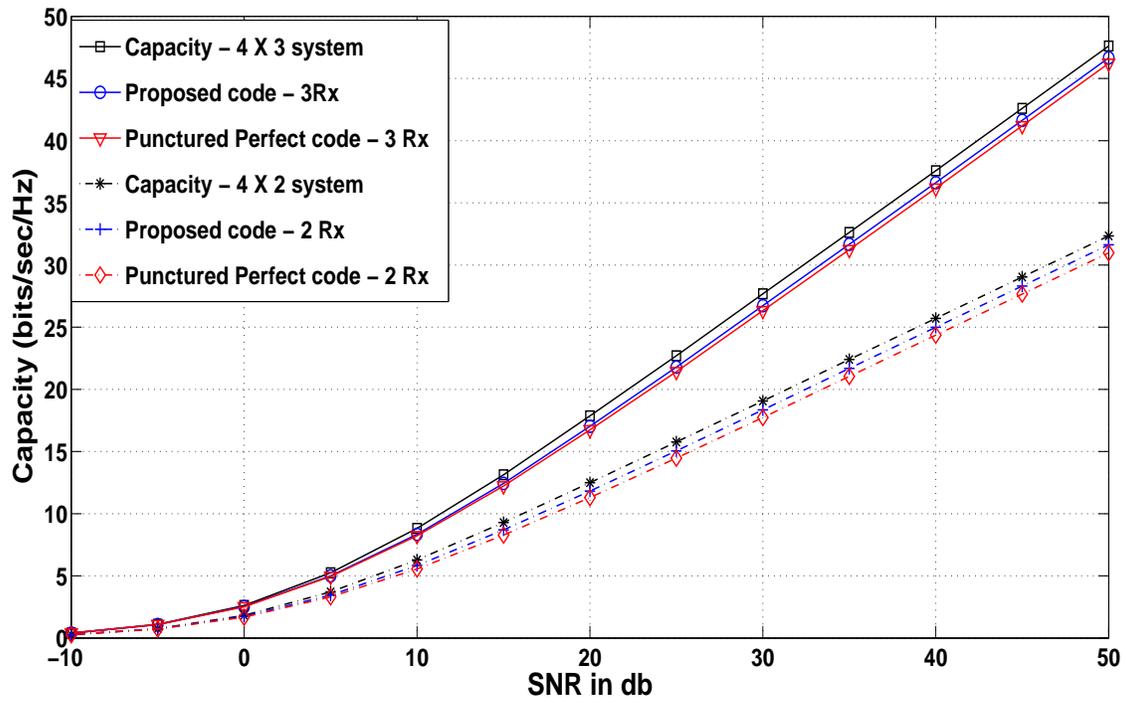}
\caption{Ergodic capacity Vs SNR for codes for $4 \times 2$ and $4 \times 3$ systems}
\label{fig_cap}
\end{figure*}

\begin{figure*}
\centering
\includegraphics[width=7in,height=4in]{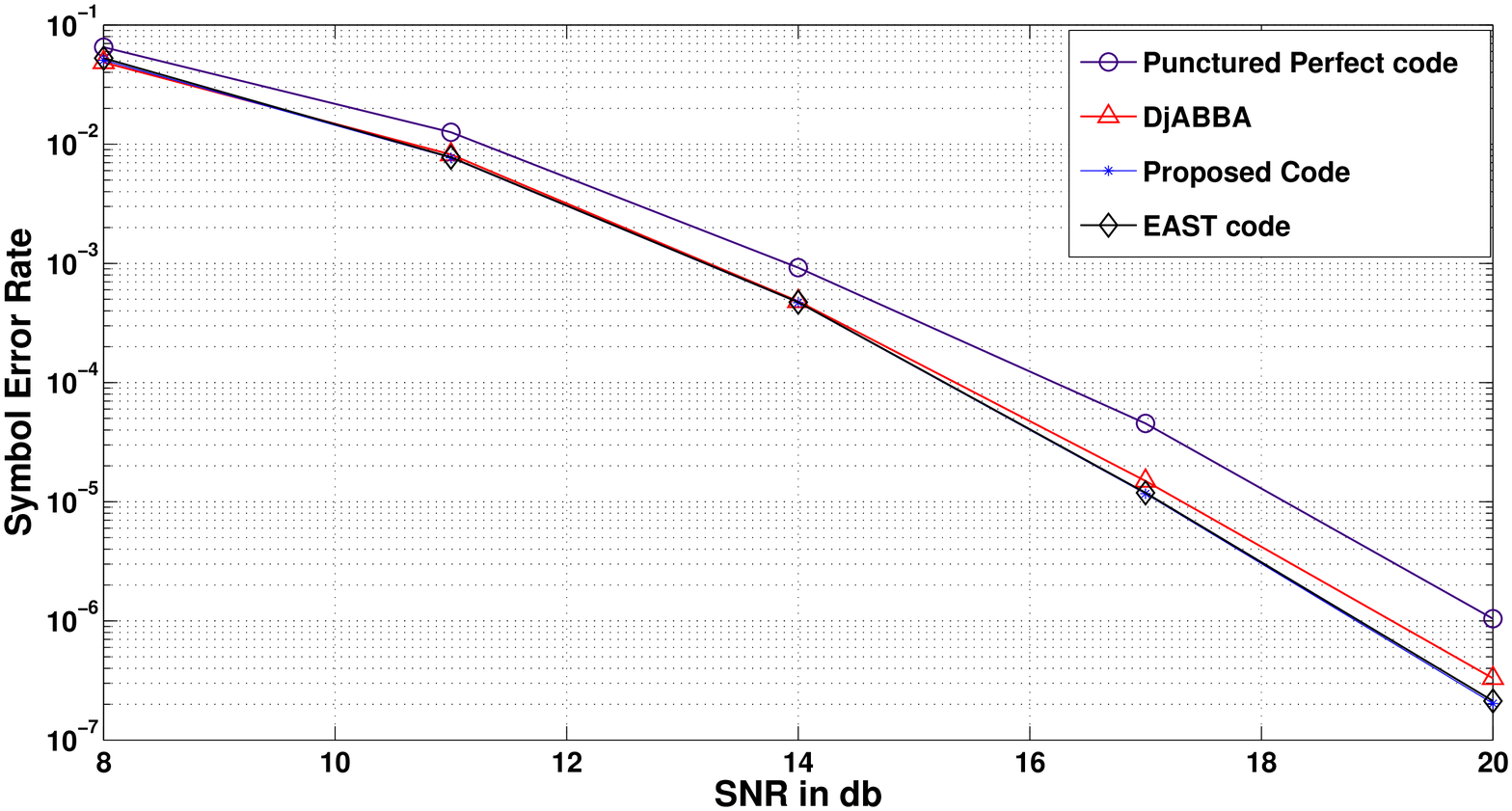}
\caption{SER performance at 4 BPCU for codes for $4 \times 2$ systems}
\label{fig1}
\end{figure*}

\begin{figure*}
\centering
\includegraphics[width=7in,height=4in]{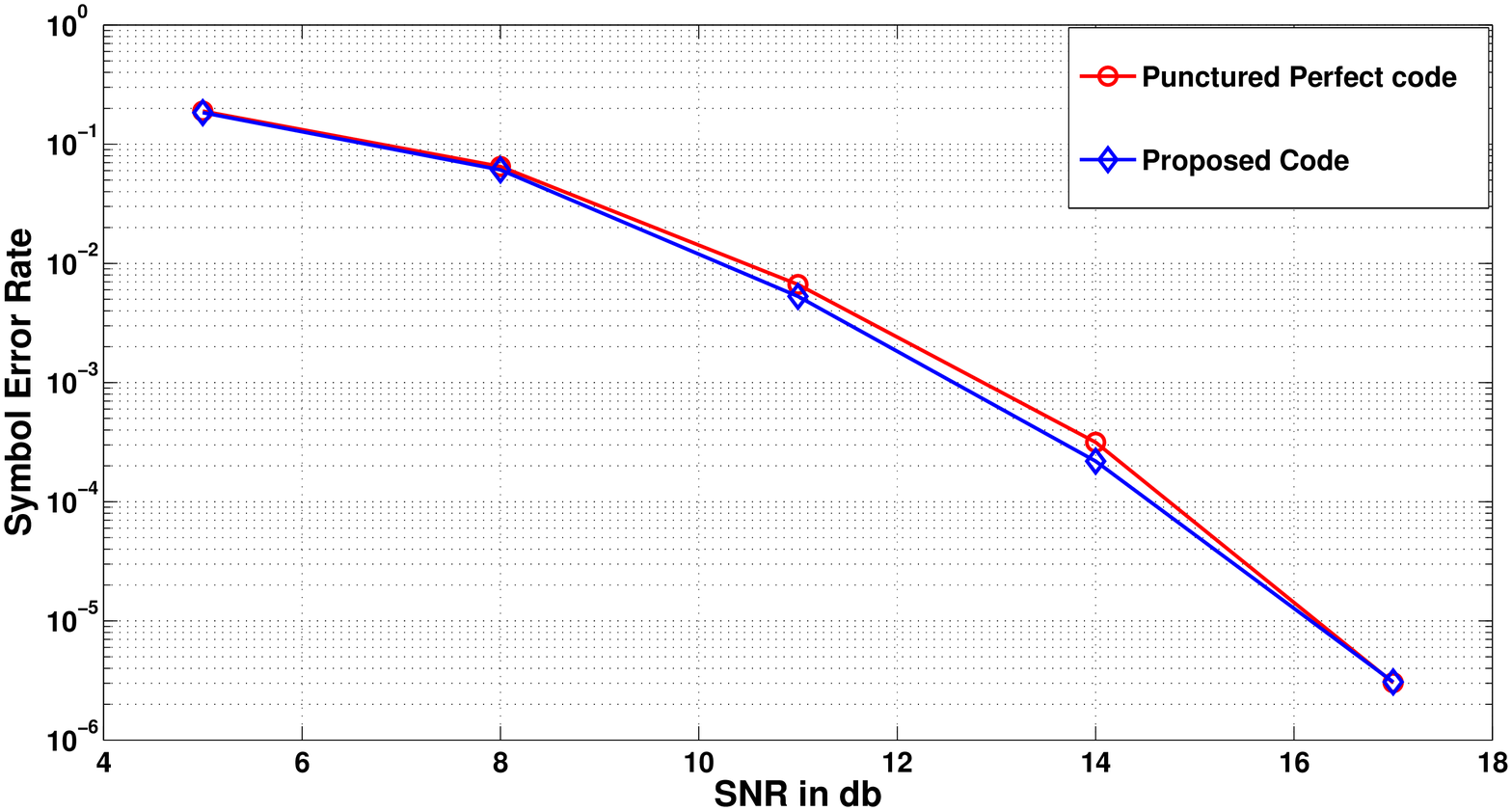}
\caption{SER performance at 6 BPCU for codes for $4 \times 3$ systems}
\label{fig2}
\end{figure*}

\begin{figure*}
\centering
\includegraphics[width=7in,height=4in]{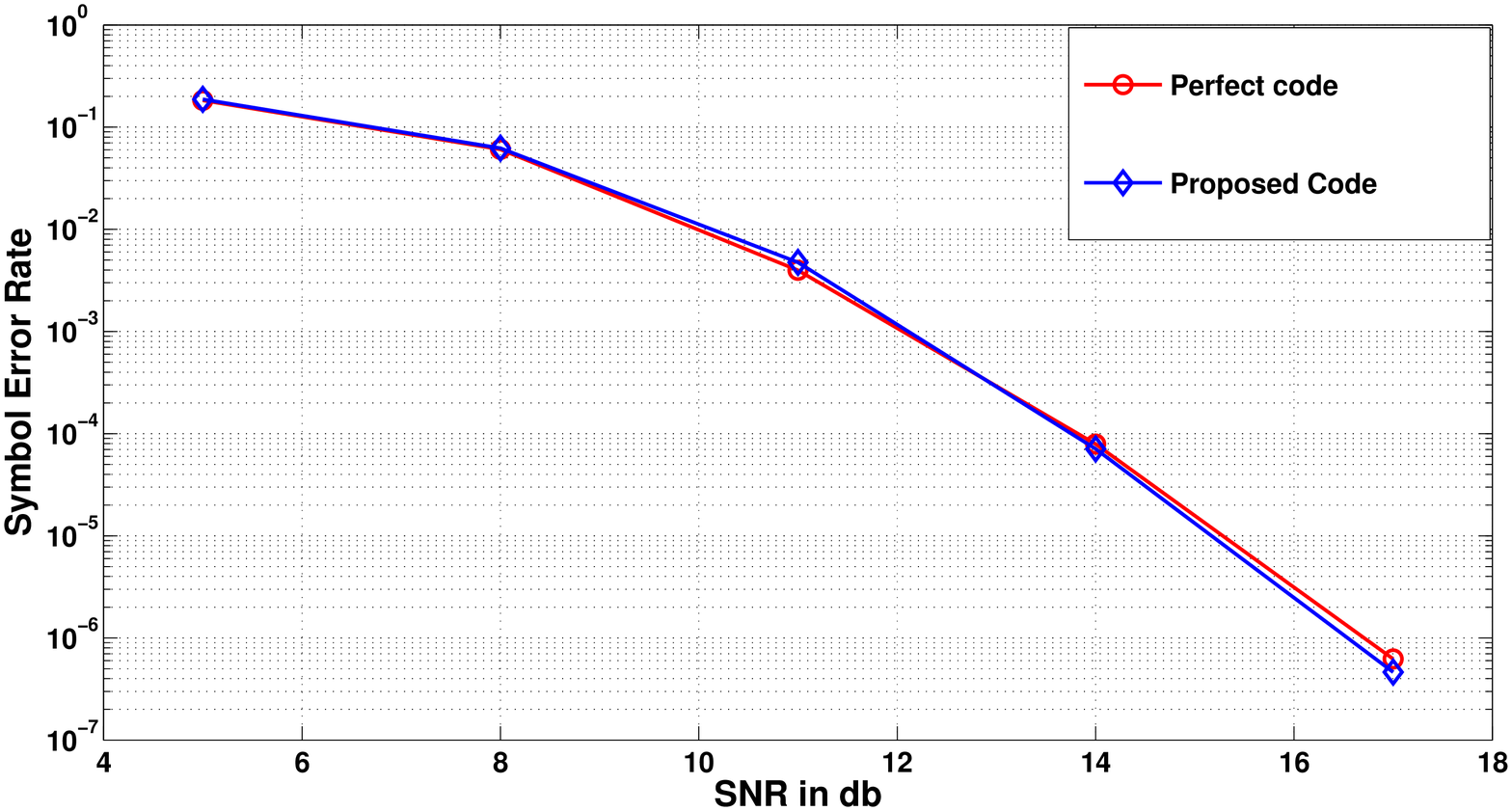}
\caption{SER performance at 8 BPCU for codes for $4 \times 4$ systems}
\label{fig3}
\end{figure*}

\section{Simulation Results} \label{sec5}
In all the simulation scenarios in this section,  we consider the Rayleigh block fading MIMO channel.
\subsection{$4 \times 2$ MIMO}
Fig. \ref{fig1} shows the plots of the symbol error rate (SER) as a function of the SNR at each receive antenna for four codes - the DjABBA code, the punctured perfect code, our code and the EAST code. Since the number of degrees of freedom of the channel is only 2, we need to use the punctured Perfect code, i.e the Perfect code with 2 of its 4 layers punctured. Our code is the one given in \eqref{rate2}. The constellation used is 4-QAM. Our code and the EAST code have the best performance. It is to be noted that the curves for our code and the EAST code coincide.
\subsection{$4 \times 3$ MIMO}
Fig. \ref{fig2} shows the plots of the symbol error rate (SER) as a function of the SNR at each receive antenna for two codes - the punctured perfect code (puncturing one of its 4 layers) and our code whose codeword is given in \eqref{rate3}. The constellation used is 4-QAM. Our punctured code has a marginally better performance than the punctured perfect code in the low to medium SNR range. As our code does not have full-diversity, at a very high
SNR, it might lose out on diversity gain.
\subsection{$4 \times 4$ MIMO}
Fig. \ref{fig3} shows the plots of the symbol error rate (SER) as a function of the SNR at each receive antenna for our code and the Perfect code. Our code nearly matches the Perfect code in performance at low and medium SNR, while at high SNR, it may lose out to the Perfect code due to the lack of full-diversity. More importantly, our code has lower ML-decoding complexity.


\section{Discussion} \label{sec6}

In this paper, we proposed a scheme to obtain a full-rate STBC for 4 transmit antennas and any number of receive antennas with reduced ML-decoding complexity. The design, although not a full-diversity code for 3 and 4 receive antennas, matches the Perfect code for 4 transmit antennas in error performance, while beating it for $4 \times 2$ MIMO systems. In terms of ergodic capacity, our proposed design has higher ergodic capacity than the punctured Perfect code for 2 and 3 receive antennas, while for $4 \times 4$ MIMO systems, it is information lossless, like the Perfect code. The scheme presented in this paper can be applied to higher number of transmit antennas to obtain similar advantages and this could provide the direction of future research.

\section*{ACKNOWLEDGEMENT}
This work was partly supported by the DRDO-IISc program on Advanced Research in Mathematical Engineering, through research grants to B. Sundar Rajan.


\begin{thebibliography}{1}
\bibitem{TJC} V.~Tarokh, H.~Jafarkhani and A.~R.~Calderbank, ``Space-Time block codes from orthogonal designs,'' \emph{IEEE Trans. Inf. Theory,} vol.45, pp.1456-1467, July 1999. Also ``Correction to ``Space-time block codes from orthogonal designs,'' \emph{IEEE Trans. Inf. Theory}, vol. 46, no.1, p.314, Jan. 2000.

\newpage

\bibitem{TiH}
O.~Tirkonen and A.~Hottinen, ``Square-matrix embeddable space-time block codes for
complex signal constellations,'' {\it IEEE Trans. Inf. Theory}, vol.48, no.2, Feb. 2002.

\bibitem{BRV}
J. C. Belfiore, G. Rekaya and E. Viterbo, ``The Golden Code: A $2\times2$ full rate space-time code with non-vanishing determinants,'' \emph{IEEE Trans. Inf. Theory}, vol. \ 51, no. \ 4, pp. \ 1432-1436, April 2005.

\bibitem{SrR_arxiv}
K. Pavan Srinath and B. Sundar Rajan, ``Low ML-Decoding Complexity, Large Coding Gain, Full-Rate, Full-Diversity STBCs for $2\times 2$ and $4 \times 2$ MIMO Systems,'' \emph{IEEE JOURNAL OF SEL. TOPICS IN SIGNAL PROCESSING}, vol.\ 3, No.\ 6, Dec. 2009.

\bibitem{john_barry1}
M. O. Sinnokrot and John Barry, ``Fast Maximum-Likelihood Decoding of the Golden Code", available online at arXiv, arXiv:0811.2201v1 [cs.IT], 13 Nov. 2008.

\bibitem{HTW}
A. Hottinen, O. Tirkkonen and R. Wichman, ``Multi-antenna Transceiver Techniques for 3G and Beyond,'' Wiley publisher, UK, 2003.

\bibitem{PGA}
J. Paredes,  A.B. Gershman,  M. Gharavi-Alkhansari, `` 	
A New Full-Rate Full-Diversity Space-Time Block Code With Nonvanishing Determinants and Simplified Maximum-Likelihood Decoding,'' \emph{IEEE Trans. Signal Processing,} vol.\ 56, No. 6, pp. \ 2461 - 2469 , Jun.\ 2008.

\bibitem{BHV}
E. Biglieri, Y. Hong and E. Viterbo, ``On Fast-Decodable Space-Time Block Codes'', \emph{IEEE Trans. Inf. Theory}, vol. \ 55, no. \ 2, pp. \ 524-530, Feb. 2009.

\bibitem{JH}
H. Jafarkhani, ``A quasi-orthogonal space-time block code,'' \emph{IEEE WCNC 2000)}, vol.\ 1, pp.\ 42-45, 2000.

\bibitem{ZS}
Zafar Ali Khan, Md., and B. Sundar Rajan, ``Single Symbol Maximum Likelihood Decodable Linear STBCs'', \emph{IEEE Trans. Inf. Theory}, vol.\ 52, No.\ 5, pp.\  2062-2091, May\ 2006.

\bibitem{Robert}
S. Sirianunpiboon, Y. Wu, A. R. Calderbank and S. D. Howard, ``Fast optimal decoding of multiplexed Orthogonal Designs," submitted to \emph{IEEE Trans. Inf. Theory}, May 2008.

\bibitem{ORBV}
F. Oggier, G. Rekaya, J. C. Belfiore and E. Viterbo, ``Perfect space time block codes,'' \emph{IEEE Trans. Inf. Theory,} vol.\ 52, No.\ 9, pp.\ 3885-3902, September 2006.

\bibitem{HaH}
B.~Hassibi and B.~Hochwald, ``High-rate codes that are
linear in space and time,'' {\it IEEE Trans. Inf. Theory}, vol.48, no.7, pp.1804-1824, July 2002.

\bibitem{anti_matric}
Daniel B. Shapiro and Reiner Martin, ``Anticommuting Matrices", \emph{The American Mathematical Monthly}, Vol.\ 105, No.\ 6(Jun. -Jul., 1998), pp.\ 565-566.

\bibitem{YGT} C.Yuen, Y.L. Guan and T.T. Tjhung, ``Quasi-orthogonal STBC with minimum decoding complexity,'' \emph{IEEE Trans. Wireless Comm.}, Vol. 4, No. 5, pp. 2089 - 2094, Sop. 2005.

\bibitem{sanjay}
Sanjay Karmakar and B. Sundar Rajan, ``Maximum-rate, Minimum-Decoding-Complexity STBCs from Clifford Algebras," submitted to \emph{IEEE Trans. Inf. Theory}, available online at arXiv, arXiv:0712.2371.


\bibitem{TSC}
 V.Tarokh, N.Seshadri and A.R Calderbank,"Space time codes for high date rate wireless communication : performance criterion and code construction'',
\emph{IEEE Trans. Inf. Theory}, vol.\ 44,  pp. 744 - 765, 1998.

\bibitem{ViB}
Emanuele Viterbo and Joseph Boutros, ``Universal lattice code decoder for fading channels'', \emph{IEEE Trans. Inf. Theory.}, vol.\ 45, no.\ 5, pp.\ 1639-1642, Jul.\ 1999.

\bibitem{barry}
M. O. Sinnokrot, John R. Barry and V. K. Madisetti, ``Embedded Alamouti Space-Time Codes for High Rate and Low Decoding Complexity", \emph{IEEE Asilomar 2008}.
\bibitem{JJK}
Jian-Kang Zhang, Jing Liu, Kon Max Wong, ``Trace-Orthonormal Full-Diversity Cyclotomic Space Time Codes,'' \emph{IEEE Trans. Signal Processing}, vol.\ 55, no.\ 2, pp.\ 618-630, Feb. 2007.

\end{thebibliography}
\end{document}